\documentclass[11pt,letter]{article}
\usepackage[hidelinks]{hyperref} 
\usepackage{fullpage}
\usepackage[utf8x]{inputenc}
\usepackage{amsthm}
\usepackage{todonotes,wrapfig,float,graphicx,amssymb,textcomp,array,amsmath}
\usepackage{enumerate,enumitem}
\usepackage{multirow}
\usepackage{tabularx}
\usepackage{color,xcolor}
\usepackage{enumitem}
\usepackage{float}
\usepackage{nicefrac}

\allowdisplaybreaks

\definecolor{mycolor}{rgb}{0, 0, 0}
\newcommand{\rev}[1]{{\color{mycolor} #1}}

\usepackage[titletoc,title]{appendix}

\usepackage{lineno}
\DeclareMathOperator{\PR}{\mathrm{Pr}}
\def\Re{\mathbb R}

\newcommand{\E}{\mathcal{E}}
\newcommand{\Y}{\mathcal{Y}}

\def\Hyp{\mathcal H}
\newcommand{\D}{\mathcal{D}}

\def\Re{\mathbb R}
\def\N{\mathbb N}
\def\EXP{\mathbb{E}}
\def\PR{\mathrm{Pr}}

\title{Polychromatic Coloring of Tuples in Hypergraphs\thanks{A preliminary version in the 41st International Symposium on Computational Geometry (SoCG 2025).}}

\author{
Ahmad Biniaz\thanks{School of Computer Science, University of Windsor, Canada, \texttt{abiniaz@uwindsor.ca}. Research supported in part by NSERC.}
\and Jean-Lou De Carufel\thanks{School of Electrical Engineering and Computer Science, University of Ottawa, Ottawa, Canada, \texttt{jdecaruf@uottawa.ca}. Research supported in part by NSERC.} 
\and  Anil Maheshwari\thanks{School of Computer Science, Carleton University, Ottawa, Canada, \texttt{anil@scs.carleton.ca}. Research supported in part by NSERC.}
\and Michiel Smid\thanks{School of Computer Science, Carleton University, Ottawa, Canada, \texttt{michiel@scs.carleton.ca}. Research supported in part by NSERC.}
\and Shakhar Smorodinsky\thanks{Institute for the Theory of Computing, Faculty of Computer and Information Science, Ben-Gurion University of the Negev, Israel, \texttt{shakhar@bgu.ac.il}. Partially supported by Grant 1065/20 from the Israel Science Foundation, by the United States – Israel Binational Science Foundation (NSF-BSF grant no. 2022792) and by ERC grant no.~882971 "GeoScape" and by the Erd{\H o}s Center.} 
\and  Milo\v{s} Stojakovi\'{c}\thanks{Department of Mathematics and Informatics, University of Novi Sad, Novi Sad, Serbia, \texttt{milosst@dmi.uns.ac.rs}. Partly supported by Ministry of Science, Technological Development and Innovation of Republic of Serbia (Grants 451-03-33/2026-03/200125 \& 451-03-34/2026-03/200125). Partly supported by Provincial Secretariat for Higher Education and Scientific Research, Province of Vojvodina (Grant~003065240 2026 09418 002 000 000 001 01 001).}
 }

\date{}
\newtheorem{lemma}{Lemma}
\newtheorem{corollary}{Corollary}

\newtheorem{theorem}{Theorem}

\newtheorem*{problem*}{Problem}
\newtheorem*{claim*}{Claim}
\newtheorem*{invariant*}{Invariant}

\newtheorem{definition}{Definition}

\definecolor{mycolor}{rgb}{0, 0, 0}
\newcommand{\ahmad}[1]{{\color{black} #1}}
\newcommand{\milos}[1]{{\color{black} #1}}

\newcommand{\etal}{{et~al.}}

\begin{document}
\maketitle
\begin{abstract}
A hypergraph $H$ consists of a set $V$ of vertices and a set $E$ of hyperedges that are subsets of $V$. A $t$-\emph{tuple} of $H$ is a subset of $t$ vertices of $V$. A $t$-tuple $k$-\emph{coloring} of $H$ is a mapping of its $t$-tuples into $k$ colors. A coloring is called $(t,k,f)$-\emph{polychromatic} if each hyperedge of $E$ that has at least $f$ vertices contains tuples of all the $k$ colors. Let $f_H(t,k)$ be the minimum $f$ such that $H$ has a $(t,k,f)$-polychromatic coloring. For a family of hypergraphs $\mathcal{H}$ let $f_{\Hyp}(t,k)$ be the maximum $f_H(t,k)$ over all hypergraphs $H$ in $\Hyp$.
Determining $f_{\Hyp}(t,k)$ has been an active research direction in recent years.
This is challenging even for $t=1$. We present several new results in this direction for $t\ge 2$.
\begin{itemize}
\item \rev{Let $\Hyp$ be the family of hypergraphs $H$ that is obtained by taking any set $P$ of points in $\Re^d$, setting $V:=P$ and $E:=\{d\cap P\colon d\text{ is a ball in $\Re^d$}\}$. We prove that $f_\Hyp(\lfloor\frac{d+3}{2}\rfloor,k)\le \left(\frac{5ed^3}{4}\right)^k$, where $e$ is the Euler's number. For $d=2$ we show that $f_\Hyp(2,k)\le 3.7^k$, that is, the pairs of points (2-tuples) can be $k$-colored such that any disk containing at least $3.7^k$ points has pairs of all colors.} 
    
     \item For the family $\Hyp$ of hypergraphs that are defined by grid vertices and axis-parallel rectangles in the plane, we show that $f_{\Hyp}(2,k)\leq \sqrt{ck\ln k}$ for some constant $c$. We then generalize this to higher dimensions, to other shapes, and to tuples of larger size.
    \item     
    For the family $\Hyp$ of shrinkable hypergraphs of VC-dimension at most $d$ we prove that $   f_\Hyp(d{+}1,k) \leq c^k$ for some constant $c=c(d)$. Towards this bound, we obtain a result of independent interest: Every hypergraph with $n$ vertices and with VC-dimension at most $d$ has a $(d{+}1)$-tuple $T$ of depth at least $\frac{n}{c}$, i.e., any hyperedge that contains $T$ also contains $\frac{n}{c}$ other vertices. We also present analogous bounds for coloring pairs of points with respect to pseudo-disks in the plane.
    \item For the relationship between $t$-tuple coloring and vertex coloring in any hypergraph $H$ we establish the inequality $\frac{1}{e}\cdot tk^{\frac{1}{t}}\le f_H(t,k)\le f_H(1,tk^{\frac{1}{t}})$. For the special case of $k=2$, referred to as the bichromatic coloring, we prove that $t+1\le f_H(t,2)\le\max\{f_H(1,2), t+1\}$; this improves upon the previous best known upper bound. 

    \item We study the relationship between tuple coloring and epsilon nets. In particular we show that if $f_H(1,k)=O(k)$ for a hypergraph $H$ with $n$ vertices, then for any $0<\epsilon<1$ the $t$-tuples of $H$ can be partitioned into $\Omega\left((\frac{\epsilon n}{t})^t\right)$ $\epsilon$-$t$-nets. This bound is tight when $t$ is a constant. 
\end{itemize}
\end{abstract}

\section{Introduction}
Polychromatic coloring and cover-decomposition of hypergraphs have rich backgrounds. They have been studied for both abstract and geometric hypergraphs as early as 1980 (see, e.g.,~\cite{Pach-1980,Pach1986,pach2005indecomposable,PachT07,pach2007decomposition}) and are still active research areas (see, e.g.,~\cite{AckKesPalvo2021,AckermanKV16,AckermanKV17,BollobasPRS11,BollobasPRS13,CardinalKMPUV23,CardinalKMU14-b,CardinalKMU14,Planken2024}). 
Perhaps the study of these topics gained more attention after a question of Pach~\cite{Pach-1980} in~1980:\footnote{This is a dual form of the original question that is stated in terms of \rev{$m$-fold coverings decomposable into $k$ many 1-fold coverings}; see \cite{PachPalvolgyi}.} 

\begin{minipage}
{.9\columnwidth}
\vspace{8pt}\emph{Does there exist a constant $f$ such that every finite set of points in the plane can be colored with two colors in such a way that any unit disk that contains $f$ points contains points of both colors?}\vspace{8pt}
\end{minipage} %at least for any finite set of points in the plane there exists a constant $m$ such that if any convex body $D$ there exists a constant $m=m(D)$ such that for any finite point set $P$ in the plane it is possible to color the points in $P$ with two colors such that any translate of $D$ that contains at least $m$ points, contains points of both colors.

\noindent This question initiated the study of polychromatic coloring in geometric hypergraphs. 
Surprisingly enough, the answer to the question turned out to be negative, as it was proved in 2016 by Pach and P{\'a}lv{\"o}lgyi \cite{PachPalvolgyi}. In particular, they showed that for any natural number $f$, there is a finite set $P$ of points in the plane such that for any $2$-coloring of $P$, one can find a unit disk containing at least $f$ points of one color {\color{mycolor} and no point of the other color}. The same question can be asked for other convex bodies and for more than two colors: Given a convex body $C$ and any natural number $k\ge 2$, is there a function $f \colon \N \to \N$ such that any finite point set in the plane can be $k$-colored such that any translate of $C$ that contains at least $f(k)$ points, contains points of all the $k$ colors? For instance,
P{\'{a}}lv{\"{o}}lgyi and T{\'{o}}th \cite{Palv-Toth-2010} proved that $f(k)$ exists when $C$ is a convex polygon---it was shown later~\cite{Aloupis2010,Gibson-Varadarajan-2011} that $f(k)=O(k)$.
The negative answer of \cite{PachPalvolgyi} to the original question of Pach implies that $f(2)$ does not exist if $C$ is a disk. 
More results and details in this direction can be found in the survey article by Pach, P{\'{a}}lv{\"{o}}lgyi and T{\'{o}}th~\cite{pach2013survey}, \rev{a more recent survey on coloring geometric hypergraphs by Dam{\'{a}}sdi, Keszegh, Pach, P{\'{a}}lv{\"{o}}lgyi and T{\'{o}}th~\cite{Gabor2025}}, and also in the webpage maintained by Keszegh and P{\'a}lv{\"o}lgyi~\cite{Zoo}.

The containment of points in disks (and other convex bodies) can be represented by geometric hypergraphs, also known as range spaces. A \emph{hypergraph} $H=(V,E)$ consists of a set $V$ of vertices and a set $E$ of hyperedges that are subsets of $V$.  A (possibly infinite) set of hypergraphs is called a \emph{family} of hypergraphs.
%For a given family $\Hyp$ of a hypergraph, the question would translate to: Is there a constant $f$ such that the vertices of any hypergraph $H$ in $\Hyp$ can be $k$-colored such that any edge that contains at least $f$ vertices, contains vertices of all colors?  
Let $f_\Hyp(k)$ be the smallest natural number such that the vertices of any hypergraph $H$ in $\Hyp$ can be $k$-colored such that any hyperedge of $H$ that has at least $f_\Hyp(k)$ vertices contains points of all colors.\footnote{\rev{This is also denoted $\chi_m(k)$ and $\chi_\text{big}(k)$ in the literature.}} If no such number exists, then $f_{\Hyp}(k):=\infty$. If $f_\Hyp(2)$ exists, then $\Hyp$ is called \emph{cover-decomposable}---a term that Pach introduced for $k=2$.

Determining $f_\Hyp(k)$ is an active research direction in both abstract and geometric hypergraphs; see e.g.~\cite{AckermanKV17,BollobasPRS13,CardinalKMPUV23,CardinalKMU13,CardinalKMU14,ChekanU22} and the references therein. Geometric hypergraphs, in particular, have received special attention. These hypergraphs are typically defined by points and convex bodies in some fixed constant dimension. For instance, let $C$ be a convex body in $\Re^2$ and let $\Hyp$ be the family of all hypergraphs $H$ that is obtained by taking any set $P$ of points in $\Re^2$, setting $V:=P$ and $E:=\{C'\cap P\colon C'\text{ is a translate of }C\}$. As mentioned above, if $C$ is a disk then $f_{\Hyp}(2)$ does not exist and hence $\Hyp$ is not cover-decomposable \cite{PachPalvolgyi}, and if $C$ is a convex polygon then $f_{\Hyp}(2)$ exists and $\Hyp$ is cover-decomposable \cite{Palv-Toth-2010}. Finding $f_\Hyp(k)$ could be more challenging if we could scale or rotate the convex body. For example if $C$ is the family of axis-parallel squares then $f_{\Hyp}(2)\le 215$~\cite{AckermanKV17}, and if $C$ is the family of axis-parallel rectangles then $f_\Hyp(2)$ does not exist~\cite{ChenPST08,Chen2009}.

A natural question thus arises: For hypergraphs that are not cover-decomposable with respect to the coloring of vertices, such as those defined by disks and axis-parallel rectangles, can we prove a positive result when we color subsets rather then the vertices?  This leads to a natural generalization of the problem: color pairs, triples, or tuples of points. This generalization is referred to as {\em tuple coloring}.  
Let $H=(V,E)$ be a hypergraph. For a natural number $t$, a $t$-\emph{tuple} of $H$ is a subset of $t$ vertices of $V$. The number of $t$-tuples of $H$ is ${|V|\choose t}$. A $t$-tuple $k$-\emph{coloring} of $H$ is mapping of its $t$-tuples into $k$ colors. 
For $t=1$, this is equivalent to vertex coloring.
A coloring is called $(t,k,f)$-\emph{polychromatic} if each hyperedge of $E$ that has at least $f$ vertices contains tuples of all the $k$ colors. 
Let $f_H(t,k)$ denote the least integer $f$ such that $H$ has a $(t,k,f)$-polychromatic coloring. For a hypergraph family $\mathcal{H}$ let $f_{\Hyp}(t,k)$ denote the maximum $f_H(t,k)$ over all hypergraphs $H$ in $\Hyp$. In other words, $f_{\Hyp}(t,k)$ is the minimum number for which every hypergraph in $\mathcal{H}$ has a $(t,k,f_{\Hyp}(t,k))$-polychromatic coloring. 
For the standard vertex coloring, where $t=1$, the value $f_{\Hyp}(1,k)$ is equal to $f_{\Hyp}(k)$ as defined above. With this definition, $\Hyp$ is cover-decomposable if $f_\Hyp(1,2)$ exists. The family $\Hyp$ is called $t$\emph{-cover-decomposable}, for $t\ge 2$, if $f_\Hyp(t,2)$ exists. 

\rev{For a hypergraph to have a $(t,k,f(t,k))$-\emph{polychromatic} coloring, every hyperedge with at least $f(t,k)$ vertices must have $t$-tuples of all the $k$ colors (by the definition)}. This implies that 
$\binom{{f(t,k)}}{t}  \geq k$. \rev{Hence, $f(t,k)\ge t+1$ for any $k\ge 2$.} By the standard upper bound for the binomial coefficient, we have $\binom{{f(t,k)}}{t}  \leq \left(\frac{e\cdot f(t,k)}{t}\right)^t$. Combining the two inequalities yields the lower bound  
\begin{equation}\label{lowerbound} f(t,k)\ge \frac{1}{e}\cdot t k^{\frac{1}{t}}.\end{equation}

\subsection{Other Related Results}
There is a rich literature for vertex-coloring (i.e.~when $t=1$) of geometric hypergraphs. In these hypergraphs, the vertices are usually defined by points and the hyperedges are defined by geometric shapes  such as disks, rectangles, squares, bottomless rectangles, strips, quadrants, triangles, octants, and halfplanes; see, for example  \cite{AckermanKV17,Aloupis2011,Aloupis2010,Asinowski2013,CardinalKMPUV23,ChekanU22,Chen2009,Gabor2025, Gibson-Varadarajan-2011,KeszeghP15,PachPalvolgyi,Palv-Toth-2010} and a recent article by Planken and Ueckerdt~\cite{Planken2024} that gives a summary of the results.

For $t\ge 2$, coloring $t$-tuples in hypergraphs is closely related to Ramsey-type problems (see \cite{Mubayi2020} for Ramsey problems in hypergraphs). In particular, it follows from Ramsey's theorem that for any two natural numbers $t$ and $d$, where $t\leq d$, the family of complete $d$-uniform hypergraphs is not $t$-cover-decomposable, i.e., $f_\Hyp(t,2)$ is unbounded.
%\todoin{Ahmad: I added the constraints $t\le d$ and being complete?}

A recent work of Ackerman, Keszegh, and P{\'a}lv{\"o}lgyi \cite{AckKesPalvo2021} is devoted to determining $f_\Hyp(t,k)$ in geometric hypergraphs that are defined by pseudo-disks, bottomless rectangles, halfspaces, and boxes. 
In particular, for the family $\Hyp$ of hypergraphs that are defined by points and axis-parallel boxes in dimension $d\ge 2$ they show that $f_\Hyp(t,k)\le k^{2^{d-1}}+t-1$ when $t\ge 2$ (to better appreciate this bound we shall recall that $f_\Hyp(t,k)$ is unbounded when $t=1$, as shown in \cite{Chen2009}). From this result it follows that for the family $\Hyp$ of hypergraphs that are defined by points and homothets\footnote{A homothet of an object is obtained by translating and \rev{positively scaling} the object.} of a fixed convex polytope with $h$ facets in some dimension it holds that  $f_\Hyp(t,k)\le k^{2^{h-1}}+t-1$ \cite{AckKesPalvo2021}.

\ahmad{
\rev{An alternative version of the problem is to color only $t$-tuples that are hyperedges of the hypergraph. Under this version the upper bounds on $f_\Hyp(t,k)$ are generally larger, because fewer $t$-tuples are available to color. Ackerman~\etal~\cite{AckKesPalvo2021} also studied this version of the problem. In particular they studied this version for $t=2$, which is equivalent to coloring the edges of the so-called {\em Delaunay graph}.}\footnote{\rev{The Delaunay graph, over a point set $P$ and a family $\cal{F}$ of regions in the plane, has $P$ as its vertices and has an edge between two points $p$, $q$ if and only if there is a region in $\cal{F}$ that contains both $p$ and $q$, but no other point of $P$.}} They show that the edges of a Delaunay graph, that is defined by pseudo-disks, can be 4-colored such that every pseudo-disk containing at 
least three points contains edges of two different colors. 
For Delaunay graphs that are defined by bottomless rectangles, they show a 2-coloring of the edges such that every bottomless rectangle containing at 
least four points contains edges of two different colors (the lower bound four on the number of points is tight).
For Delaunay graphs that are defined by halfplanes, they show how to 2-color the edges such that every halfplane containing at 
least three edges contains edges of two different colors (the lower bound three for edges is tight).
}

\subsection{Preliminaries}
When the family or hypergraphs or a particular hypergraph is clear from the context we may drop the subscript and simply write $f(t,k)$ for $f_\Hyp(t,k)$ and $f_H(t,k)$. First we define the notion of shrinkability which is a common property of many geometric hypergraphs.

\begin{definition}[Shrinkability]\label{def:shrinkability}
    A hypergraph $H$ is shrinkable if for every hyperedge $e\in H$ and for every integer $1 \leq i \leq |e|$ there exists a hyperedge $e'$ in $H$ such that $e' \subseteq  e$ and $|e'|=i$.
\end{definition}

The VC-dimension of a hypergraph is a measure of its complexity, and it plays a central role in statistical learning, computational geometry, and other areas of computer science and
combinatorics (see, e.g.,~\cite{AHW87,BEHW89,Matousek04,MV18}).
\begin{definition}[VC-dimension]
     The VC-dimension of a hypergraph $H=(V,E)$ is the size of the largest subset $S\subseteq V$ such that for every subset $B \subseteq S$ there is a hyperedge $e \in E$ where $e \cap S=B$. Such a subset $S$ is said to be ``shattered'' by $H$. 
\end{definition}

The following inequality, usually referred to as the Sauer-Shelah-Perles lemma \cite{Buzaglo2013,Sauer1972,Shelah1972}, relates the number of vertices and hyperedges of a hypergraph of VC-dimension $d$: \[|E| \leq \sum_{i=0}^{d} {\binom{|V|}{i}}.\]

%We say that a hypergraph $H=(V,E)$ is {\em $t$-covered} if every $t$-tuple of $V$ is contained in some hyperedge of $E$.
For a hypergraph $H=(V,E)$ and a subset $X\subseteq V$ the {\em projection} of $H$ to $X$ is the hypergraph $H[X]=(X,\{e\cap X\colon e\in E\})$. The VC-dimension of $H[X]$ is, at most, the VC-dimension of $H$ because any subset that is shattered by $H[X]$ is also shattered by $H$.

Let $P$ be a set of points in the plane. The {\em depth} of a pair $\{x,y\}$ of points in $P$, denoted by $d_P(x,y)$, is defined as the maximum integer $i$ such that any disk containing $x$ and $y$ contains at least $i$ other points of $P$. Notice that $0 \leq d_P(x,y) \leq |P|-2$. The existence of pairs of large depth (also known as deep pairs) is studied in \cite{BaranySchmerlSidneyUrrutia:PointsBallsEuclideanSpace-1989,EdelsbrunnerHasanSeidelShen:CirclesEncloseManyPoints-1989,Hayward:CircleContainmentProblem-89,HaywardRappaportWenger:ExtremalResultsCirclesContainingPoints-1989,Neumann-LaraU88,Ramos2009}. The notion of depth can be generalized to tuples of points and to higher dimensions as the maximum number $i$ such that any ball containing a deep tuple also contains $i$ other points. We extend the notion of depth further into hypergraphs. 

\begin{definition}[Depth]
   \label{def:depth-bounded-VC}
    Let $H=(V,E)$ be a hypergraph and $S$ be a subset of $V$. The depth $d_H(S)$ of $S$ is defined as the maximum integer such that for every hyperedge $e\in E$ that contains $S$, we have $|e{\setminus}S|\ge d_H(S)$. If no hyperedge contains $S$ then $d_H(S):= |V{\setminus}S|.$
    \end{definition}
Note that if $S$ has depth $i$, then every hyperedge that contains $S$ contains at least $i$ other vertices.
The notion of depth is monotone with respect to projections, in the sense that for any $X\subseteq V$, we have $d_{H[X]}(S) \leq d_H(S)$. 

Let $H$ be a hypergraph and $0<\epsilon<1$ be a real number. A subset $S$ of vertices is called an \emph{$\epsilon$-net} for $H$ if any hyperedge $e$  of size at least $\epsilon |V|$ contains a vertex from $S$. This notion was introduced by Haussler and Welzl~\cite{HausslerW86,hw-ensrq-87}. 
It was then generalized by Mustafa and Ray~\cite{MR17} to tuples and appears under \emph{$\epsilon$-$t$-nets}; see also~\cite{AlonJKSY20,AJKSY22,DuttaGJM17,DGJM19}.

\begin{definition}[$\epsilon$-$t$-Net]
	Let 
	$H=(V,E)$ be a hypergraph, $\epsilon \in (0,1)$ be a real number, and $t$ be a natural number. A subset $N$ of $t$-tuples of $V$ is called an $\epsilon$-$t$-net if any hyperedge $e$ in $E$ of size at least $\epsilon |V|$ contains at least one of the $t$-tuples in $N$.
\end{definition}

%For a set $T$ and an integer $k$ we use the notation $\binom T k$ to denote the set of all $k$-element subsets of $T$.

\subsection{Our Contributions}
We present upper bounds on $f(t,k)$ for several families of geometric hypergraphs and hypergraphs of bounded VC-dimension. Some of our bounds are new, and some improve over previously known bounds. 
These bounds are interesting because, for abstract hypergraphs, $f(t,k)$ is unbounded even if the hypergraph is uniform and $k=2$ (due to Ramsey's theorem; see \cite{Mubayi2020}). 
We also study the relationship between vertex-coloring, tuple-coloring, cover-decomposability, and epsilon $t$-nets.

Recall from the result of Pach and Pálvölgyi \cite{PachPalvolgyi} that $f(t,2)$, with $t=1$, is unbounded for hypergraphs that are defined by points and disks in the plane. We prove, however, that it is bounded for $t=2$. {\color{mycolor}A point set $P$ in dimension $d$ is said to be in {\em general position} if no $(d+2)$ points of $P$ lie on a sphere.}  \rev{First we prove a more general bound for points and balls in dimension $d$. The parameter $e$ below denotes the Euler’s number.}

\begin{theorem}
    \label{point-ball-thr}
Let $d\ge 2$ be an integer and set $t:=\lfloor\frac{d+3}{2}\rfloor$. Let $\Hyp$ be the family of all hypergraphs $H$ that can be obtained by taking a finite set $P$ of points in $\Re^d$ in general position setting $V(H):=P$ and \text{$E(H):=\{b\cap P\colon b$ is a ball in $\Re^d\}$}. Then for any natural number $k$ it holds that \rev{\[f_{\Hyp}(t,k)\leq \left(\frac{5ed^3}{4}\right)^k.\]}
\end{theorem}

To obtain this bound, we show some properties of containment of points in balls and also employ a result about the depth of tuples of points --- a well-studied topic in discrete and computational geometry~\cite{SSW08}. \rev{For $d=2$ this implies that $f_{\Hyp}(2,k)\leq (10e)^k\approx 27.18^k$. In other words, for any finite set $P$ of points in the plane, the $\binom{|P|}{2}$ pairs of points can be $k$-colored such that any disk containing at least $27.18^k$ points must also contain a pair (of points) from each of the $k$ color classes. In the following theorem we improve this bound by using better results for the depth of pairs of points in the plane~\cite{EdelsbrunnerHasanSeidelShen:CirclesEncloseManyPoints-1989,Ramos2009}. }

\begin{theorem}
    \label{thm:disks-plane}
Let $\Hyp$ be the family that contains any hypergraph $H$ that can be obtained by taking a finite set $P$ of points in $\Re^2$ in general position setting $V(H):=P$ and $E(H):=\{d\cap P\colon d$ is a disk in $\Re^2\}$. Then for any natural number $k$ it holds that $f_{\Hyp}(2,k)\leq 3.7^{k}$.
\end{theorem}

For points and axis-parallel rectangles in the plane, we know from \cite{Chen2009} that $f(t,2)$ is unbounded when $t=1$. However for $t=2$ a result of~\cite{AckKesPalvo2021} implies that $f(2,k)\le k^{2}+1$. We study the same problem but for grid points.
%We obtain the following result.

\begin{theorem}
    \label{grid-rectangle-thr}
\rev{Let $P$ be the set of all vertices of a regular grid in $\Re^2$.} Let $\Hyp$ be the family of all hypergraphs $H$  obtained by  setting $V(H){:=}P$ and $E(H){:=}\{r\cap P:r$ is an axis-parallel rectangle in $\Re^2\}$. Then for any natural number $k\ge 2$ it holds that $f_{\Hyp}(2,k)\leq \sqrt{ck\ln k}$, \rev{for some constant $c$}.
\end{theorem}
To better appreciate this bound, \rev{recall from~\eqref{lowerbound} that $f(2,k)=\Omega(\sqrt{k})$ for any hypergraph.} We generalize this result to boxes and balls in any fixed dimension $d\ge 2$.

\begin{theorem}
    \label{grid-box-thr}
Let $\Hyp$ be the family of all hypergraphs that can be obtained by grid points and boxes (or balls) in any fixed dimension $d\ge 2$. Then for any natural numbers $t\ge 1$ and $k\ge 2$ it holds that $f_{\Hyp}(t,k)=O\left(\left(k\ln k\right)^{1/t}\right)$.
%, where the constant of proportionality in the big-Oh notation depends on $d$.
\end{theorem}

We also study hypergraphs of bounded VC-dimension.
\begin{theorem}
    \label{thm:poly-boundedVC}
    For every \rev{$d\ge 2$}, every shrinkable hypergraph $H$ with VC-dimension at most $d$ and with $n\ge 2d+2$ vertices it holds that $f_H(d+1,k)\le c^{k-1}$ for any \rev{$k\ge 2$} and \rev{$c = \rev{18}(d+1)!$}.
\end{theorem}
This result applies to many families of geometric hypergraphs. For example, hypergraphs defined by points and halfspaces, boxes, or balls in $\Re^d$ are shrinkable and have bounded VC dimensions. \rev{In particular Theorem~\ref{thm:poly-boundedVC} gives analogous bounds for Theorems~\ref{point-ball-thr} and \ref{thm:disks-plane} with larger tuple sizes.} It also applies to hypergraphs defined by points and pseudo-disks\footnote{A collection of simple closed Jordan regions for every pair of which their boundaries intersect at most twice.} in the plane because they are shrinkable \cite{Pin14} and have VC dimension at most 4 \cite{Aronov2021}.
Along with a proof of this result, we show the existence of a tuple of large depth in hypergraphs of bounded VC-dimension, which is of independent interest. 

\begin{theorem}
\label{lem:depth-boundedVC}
For every \rev{$d\ge 2$}, every hypergraph $H$ with VC-dimension  at most $d$ and with $n\ge 2d+2$ vertices  
there exists a $(d+1)$-tuple $T$ of vertices such that $d_H(T) \geq \frac{n}{c}$, for \ahmad{$c = \rev{18}(d+1)!$}.
\end{theorem}

 We also present bounds, analogous to that of Theorem~\ref{thm:poly-boundedVC} and Theorem~\ref{lem:depth-boundedVC}, for coloring pairs of points with respect to pseudo-disks in the plane.

For hypergraphs where $f(1,x)$ is linear in $x$, we prove the following asymptotic tight bounds for their tuple-coloring and for decomposition of their tuples into epsilon $t$-nets. 

\begin{theorem}
    \label{cor:poly-well-behaved}
    Let $H$ be a hypergraph for which  $f_{H}(1,x)=O(x)$, for  $x\ge 1$. 
    Then for any natural numbers $t$ and $k$ with $tk^{1/t}\ge 1$ we have
    $
    f_{H}(t,k)= \Theta(tk^{\frac{1}{t}}).$
\end{theorem}

\begin{theorem}
    \label{thm:tnet-decomposition}
    Let $H$ be a hypergraph with $n$ vertices such that $f_{H}(1,x)=O(x)$, for  $x\ge1$. Then for any $0<\epsilon < 1$ and any $t\ge 1$ the $t$-tuples of vertices of $H$ can be decomposed into $\Omega\left((\frac{\epsilon n}{t})^t\right)$ pairwise disjoint $\epsilon$-$t$-nets for $H$. This bound is asymptotically tight when $t$ is a constant.
\end{theorem}
For an implication of Theorem~\ref{cor:poly-well-behaved} consider the family $\Hyp$ of hypergraphs that are defined by points and halfplanes in $\Re^2$. From a result of  Smorodinsky and Yuditsky~\cite{SmorodinskyY10,smorodinsky2012polychromatic} we know that $f_{\Hyp}(1,x)\le 2x{-}1$, i.e., any set of points in the plane can be colored by $x$ colors such that any halfplane containing at least $2x {-} 1$ points contains points of all colors. Therefore, Theorem~\ref{cor:poly-well-behaved} implies that $f_{\Hyp}(t,k)= \Theta(tk^{\frac{1}{t}})$.
\rev{The same bound holds for the family $\Hyp$ of hypergraphs defined by points and translates of convex polygons in the plane because $f_{\Hyp}(1,x)=O(x)$ for this family~\cite{Aloupis2010,Gibson-Varadarajan-2011}. An attractive conjecture in this direction is that for any hereditary\footnote{\rev{A family of hypergraphs is hereditary if it is closed
under taking subhypergraphs and traces.}} family $\Hyp$ if $f_{\Hyp}(1,2)$ is bounded then $f_{\Hyp}(1,x)=O(x)$; see the survey~\cite{Gabor2025}.}

Finally, we study the relationship between $f(1,k)$ and $f(t,k)$ for $k=2$, which is usually referred to as bichromatic coloring. It is implied from a result of Ackerman, Keszegh, and P{\'a}lv{\"o}lgyi \cite[Proposition 4]{AckKesPalvo2021} that \rev{$f_H(t,2)$}$\le\max\{f_H(1,2), 2t-1\}$ for any family $H$ and any $t\ge 2$. We present the following improved upper bound.
\begin{theorem}
\label{thm:cover-decomposable-vs-tuples}
    For any hypergraph $H$ and any integer $t\ge 2$ it holds that\[t+1\le f_H(t,2)\le\max\{f_H(1,2), t+1\}.\]
\end{theorem}
The matching term $t+1$ in the lower and upper bounds of Theorem~\ref{thm:cover-decomposable-vs-tuples} determines the exact value of $f(t,2)$ for some classes of hypergraphs. For instance, when $H$ is defined by points and halfplanes in $\Re^2$, the result of Smorodinsky and Yuditsky~\cite{smorodinsky2012polychromatic} gives $f_H(1,2)=3$, and thus Theorem~\ref{thm:cover-decomposable-vs-tuples} implies the following corollary.

\begin{corollary}
\label{thm:half-planes}
Let $t\ge 2$ be an integer. For the family $\Hyp$ of hypergraphs that are defined by points and halfplanes in $\Re^2$, it holds that
$f_\Hyp(t,2)=t+1$.
\end{corollary}

\section{Coloring Tuples of Points with respect to Balls and Disks} 
In this section, we prove Theorem~\ref{point-ball-thr} and Theorem~\ref{thm:disks-plane}. Our proofs \rev{rely} on some simple structural properties of the containment of points in disks and balls. They also employ some results about the existence of deep tuples in point sets. {\color{mycolor}Without loss of generality, all point sets that are considered in this section are assumed to be in general position.}

\subsection{Points and Balls in $\Re^d$}

Let $\Hyp$ be the family of hypergraphs defined by points and \rev{balls in dimension $d\ge 2$. That is, $\Hyp$ contains any hypergraph $H$ that can be obtained by taking a finite set $P$ of points in $\Re^d$ setting $V(H):=P$ and $E(H):=\{b\cap P\colon\text{$b$ is a ball in $\Re^d$}\}$. In this section we prove that $f_{\Hyp}(t,k)\le  \left(\frac{5ed^3}{4}\right)^k$ for $t=\lfloor\frac{d+3}{2}\rfloor$ and any $k\ge 2$. 
We say that a ball $b$ {\em contains} a point $p$ if $p$ is in the interior or on the boundary of $b$. Alternatively, we say that $p$ is in $b$.
%We say that a set of points in $\Re^d$ is in {\em general position} if no $d+2$ points lie on the boundary of a ball in $\Re^d$. For $d=2$, this means that a set of points in the plane is in general position if no four points  lie on a circle. {\color{mycolor}Without further mentioning, we assume that point sets considered in this section are in general position.} 

One part of our proof employs the following result from \cite{SSW08};

\begin{lemma}[Smorodinsky, Sulovsk\'{y}, and Wagner \cite{SSW08}]
\label{ball-depth-lemma}
In any set of $n$ points in $\Re^d$ there is a subset $S$ of size $\lfloor\frac{d+3}{2}\rfloor$ such that any ball containing $S$ contains at least
$\frac{4n}{5ed^3}$ other points. 
\end{lemma}

The coefficient $\frac{4}{5ed^3}$ in the lemma is an improvement over the previous (smaller) constant due to B\'ar\'any et al.~\cite{BaranySchmerlSidneyUrrutia:PointsBallsEuclideanSpace-1989}. It is surprising that the bound $\lfloor\frac{d+3}{2}\rfloor$ in the lemma is strongly sharp as proved in \cite{BaranySchmerlSidneyUrrutia:PointsBallsEuclideanSpace-1989}, in the sense that, there is a set $P$ of points in $\Re^d$ (obtained on the moment curve) such that for any subset $S$ of $P$ with  $|S|<\lfloor\frac{d+3}{2}\rfloor$  there is a ball that contains $S$ but no other point of $P$. However, for any  $t>\lfloor\frac{d+3}{2}\rfloor$ one can find a subset $S$ of size $t$ such that any ball containing $S$ contains at least $\frac{4n}{5ed^3}-(t-\lfloor\frac{d+3}{2}\rfloor)$ other points. This is simply obtained by taking any subset satisfying Lemma~\ref{ball-depth-lemma} and adding arbitrary points to it.}

\rev{The following lemmas, though very simple, play important roles in our proof. We denote the boundary circle of a ball $B$ by $\partial B$. 

\begin{figure}[ht!]
	\centering
\setlength{\tabcolsep}{0in}
$\begin{tabular}{cc}
	\multicolumn{1}{m{.5\columnwidth}}{\centering\vspace{0pt}\includegraphics[width=.23\columnwidth]{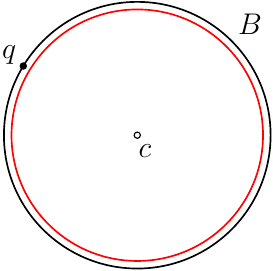}}
		&\multicolumn{1}{m{.5\columnwidth}}{\centering\vspace{0pt}\includegraphics[width=.26\columnwidth]{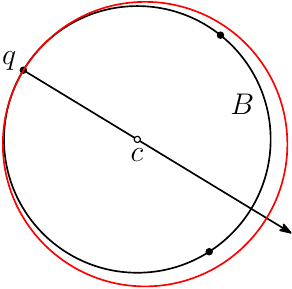}}
		\\
		(a)   &(b) 
\end{tabular}$	
	\caption{(a) Shrinking $B$ to lose $q$. (b) Enlarging $B$ along the ray from $q$ to $c$.}
\label{shrink-fig}
\end{figure}
\begin{lemma}
\label{lem:shrinking1}
    Let $P$ be a finite set of points in $\Re^d$ and let $B$ be a ball with $|B\cap P| = j$ for some $j\ge 1$. There exists a ball $B'$ such that $(B'\cap P) \subseteq B$ and $|B' \cap P|=j-1$.
\end{lemma}

\begin{proof}
Fix $B$ at its center $c$ and shrink it until a point $q\in P$ appears on $\partial B$ for the first time. If $q$ is the only point of $P$ on $\partial B$, we shrink $B$ slightly to lose $q$, as in Figure~\ref{shrink-fig}(a). Assume that there are other points of $P$ on $\partial B$. Fix $B$ at $q$ and enlarge it slightly along the ray from $q$ to $c$ such that all other points on $\partial B$ fall inside $B$ and no new point appears in $B$, as in Figure~\ref{shrink-fig}(b). Now, $q$ is the only point on $\partial B$. We fix $B$ at the center of the new ball and shrink it slightly to lose $q$.
\end{proof}

By repeatedly applying Lemma~\ref{lem:shrinking1}, one can shrink/enlarge a ball $B$ to lose 1, 2, 3, or more points of $P$. Therefore, we get the following more general lemma.

\begin{lemma}
    \label{lem:shrinking}
    Let $P$ be a finite set of points in $\Re^d$ and $B$ be a ball with $|B\cap P| = j$. For any integer $0 \leq i \leq j$ there exists a ball $B'$ such that $(B'\cap P) \subseteq B$ and $|B' \cap P|=i$.
\end{lemma}

}

\rev{

\begin{lemma}
    \label{lem:depth-monotonicity}
    Let $P$ be a finite set of points in $\Re^d$, and let $S$ be a subset of $P$.  Then for any tuple $T$ in $S$ it holds that $d_S(T)\le d_P(T)$.
\end{lemma}
\begin{proof}
    This is true because for any ball $B$ in $\Re^d$, the number of points of $P$ in $B$ is at least the number of points of $S$ in $B$. 
\end{proof}
}
We are ready for the proof of Theorem~\ref{point-ball-thr}.

\rev{
\begin{proof}[Proof of Theorem~\ref{point-ball-thr}]Fix $\lambda=\frac{5ed^3}{4}$.
Consider any hypergraph $H$ in $\Hyp$ with vertex set $P$. We color each $t$-tuple $T$ in $P$ with one of the $k$ colors $\{0,1,\ldots,k-1\}$  as follows:
\[ \text{color of }T =
  \begin{cases}
        0   & \quad \text{if } d_P(T) = 0 \\
    i+1  & \quad \text{if } \lambda^{i}\le d_P(T) < \lambda^{i{+}1} \text{ for some }i \in \{0,1,\ldots,k-3\}\\
    k-1  & \quad \text{if } d_P(T) \ge \lambda^{k-2}.
  \end{cases}
\]
We prove that any ball $B$ containing at least $\lambda^{k}$ points has $t$-tuples of all the  $k$ colors. 

By an application of Lemma~\ref{lem:shrinking} there is a ball $B_0$ such that $(B_0\cap P)\subseteq B$ and $|B_0\cap P|=t$. The $t$ points in $B_0$ form a tuple of depth 0, and hence colored 0.

For each $i\in\{0,\ldots,k-2\}$, we show that there is a $t$-tuple with color $i{+}1$ in $B$. Let $B_i$ be a ball  such that $(B_i\cap P)\subseteq B$ and $|B_i\cap P|=\lambda^{i{+}1}$; such a ball exists by Lemma~\ref{lem:shrinking} and the fact that $\lambda^{i{+}1}{<}\lambda^k\le |B\cap P|$. Let $P'$ be the set of points in $B_i$. By \rev{Lemma~\ref{ball-depth-lemma}} the set $P'$ contains a $t$-tuple $T$ such that any ball containing $T$ has at least $\frac{|P'|}{\lambda}=\lambda^i$ other points of $P'$. Thus $d_{P'}(T)\ge \lambda^i$. Since $P'\subseteq P$, by Lemma~\ref{lem:depth-monotonicity} we get $d_{P}(T)\ge \lambda^i$. On the other hand, $T$ is contained in $B_i$ which has  $\lambda^{i{+}1}$ points of $P$. Excluding the points of $T$, we get $d_{P}(T)< \lambda^{i{+}1}$. Thus, $T$ colored $i{+}1$.
\end{proof}
}

\vspace{10pt}\noindent{\bf Remark.}
The bound of Theorem~\ref{point-ball-thr} also applies to hypergraphs defined by points and halfspaces in dimension $d$ because a halfspace can be seen as a very large ball.

\subsection{Points and Disks in $\Re^2$}
\rev{For the family $\Hyp$ of hypergraphs defined by points and disks in the plane (i.e. $\Re^2$) the result of the previous section gives $f_{\Hyp}(2,k)\le  \left(\frac{5e2^3}{4}\right)^k\approx 27.18^k$. That is, the pairs of points can be colored by $k$ colors such that any disk containing at least $27.18^k$ points contains pairs with all the $k$ colors. In this section we show how this bound can be improved by using better results for the depth of pairs of points in the plane. With the following result of Edelsbrunner, Hasan, Seidel, and Shen~\cite{EdelsbrunnerHasanSeidelShen:CirclesEncloseManyPoints-1989} (instead of Lemma~\ref{ball-depth-lemma}) one can show that $f_{\Hyp}(2,k)\le 4.7^k$.

\begin{lemma}[Edelsbrunner \etal~\cite{EdelsbrunnerHasanSeidelShen:CirclesEncloseManyPoints-1989}]
\label{Edelsbrunner-lemma}
In any set of $n$ points in the plane, there is a pair of points such that any circle through them contains at least $\frac{n}{4.7}$ other points. 
\end{lemma}
}
Ramos and Via\~{n}a \cite{Ramos2009} showed a stronger version of Lemma~\ref{Edelsbrunner-lemma}. We use this version to improve the upper bound further.

\begin{lemma}[Ramos and Via\~{n}a \cite{Ramos2009}]
\label{Ramos-lemma}
In any set of $n$ points in the plane, there is a pair of points such that any circle through them has, both inside and outside, at least $\frac{n}{4.7}$
 points (excluding the two points themselves). 
\end{lemma}
This result is obtained by using the duality (of circles in $\Re^2$ and planes in $\Re^3$) and some known results on facets of points in $\Re^3$. The pair that is obtained by Lemma~\ref{Ramos-lemma} has the property that any circle through them has at least $\frac{n}{4.7}$ and at most $n{-}\frac{n}{4.7}{-}2$ points inside. Combining this with the notion of depth, we get the following corollary.

\begin{corollary}
\label{Ramos-cor}
    In any set of $n$ points in the plane  there is a pair $\{x,y\}$ such that \[\frac{n}{4.7}\le d_P(x,y)<\frac{3.7n}{4.7}.\]
\end{corollary}

The following lemma is the last ingredient we need for the proof of Theorem~\ref{thm:disks-plane}.

\begin{lemma}
    \label{two-point-disk-lemma}
    Let $D$ be a disk and $x,y$ be two points in the interior of $D$. There exists a disk that is contained in the interior of $D$ such that its boundary passes through $x$ and $y$.
\end{lemma}

\begin{proof}
Fix $D$ at its center $c$ and shrink it until a point, say $x$, appears on $\partial D$ for the first time. Then fix $D$ at $x$ and shrink it along the segment $cx$ until $y$ appears on $\partial D$. This disk satisfies the constraints of the lemma.
\end{proof}

\begin{proof}[Proof of Theorem~\ref{thm:disks-plane}]
Consider any hypergraph $H$ in $\Hyp$ with vertex set $P$. We color each pair $\{x,y\}$ of points in $P$ with one of the $k$ colors $\{0,1,\ldots,k-1\}$  as follows:
\[ \text{color of }\{x,y\} =
  \begin{cases}
        0   & \quad \text{if } d_P(x,y) = 0 \\
    i+1  & \quad \text{if } 3.7^i\le d_P(x,y) < 3.7^{i{+}1} \text{ for some }i \in \{0,1,\ldots,k-3\}\\
    k-1  & \quad \text{if } d_P(x,y) \ge 3.7^{k{-}2}.
  \end{cases}
\]
We prove that any disk $D$ containing at least $3.7^k$ points has pairs of all the  $k$ colors. 

By an application of Lemma~\ref{lem:shrinking} there is a disk $D_0$ such that $(D_0\cap P)\subseteq D$ and $|D_0\cap P|=2$. The two points in $D_0$ have depth 0 and thus are colored 0. 

\begin{figure}[ht!]
	\centering
\includegraphics[width=.28\columnwidth]{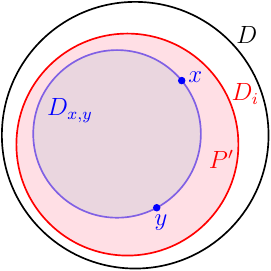}
	\caption{Illustration of the proof of Theorem~\ref{thm:disks-plane}.}
\label{depth-fig}
\end{figure}

For each $i\in\{0,\ldots,k-2\}$, we show that there is a pair of points with color $i{+}1$ in $D$. Let $D_i$ be a disk  such that $(D_i\cap P)\subseteq D$ and $|D_i\cap P|=4.7\cdot 3.7^i$ (such a disk exists by Lemma~\ref{lem:shrinking} and the fact that $4.7\cdot 3.7^i<3.7^k\le |D\cap P|$). Let $P'$ be the set of points in $D_i$. See Figure~\ref{depth-fig}. By \rev{Lemma~\ref{Ramos-lemma}} the set $P'$ contains a pair $\{x,y\}$ such that any disk through them has at least $\frac{|P'|}{4.7}=3.7^i$ and at most $\frac{3.7|P'|}{4.7}-2=3.7^{i{+}1}-2$ points of $P'$ inside. Thus $3.7^i\le d_{P'}(x,y)< 3.7^{i{+}1}$. We argue that these bounds also hold with respect to $P$, that is $3.7^i\le d_{P}(x,y)< 3.7^{i{+}1}$. The lower bound is implied from Lemma~\ref{lem:depth-monotonicity} because $d_{P'}(x,y)\le d_P(x,y)$. To verify the upper bound, we take an arbitrary disk $D_{xy}$ through $x$ and $y$ in the interior of $D_i$ (such a disk exists by Lemma~\ref{two-point-disk-lemma}). The disk $D_{xy}$ contains at most $3.7^{i{+}1}$ points of $P'$. Since $D_{xy}$ does not contain any point outside $D$, it has at most $3.7^{i{+}1}$ points of $P$. This verifies the upper bound on $d_P(x,y)$. Therefore, the pair $\{x,y\}$ has color $i{+}1$.
\end{proof}

\vspace{10pt}\noindent{\bf Remark.} To assist the readability of the proof of Theorem~\ref{thm:disks-plane} we avoided rounding when applying Lemma~\ref{lem:shrinking}; this has a negligible impact on the proof. The disk $D_i$ would be chosen to have $\lceil4.7\cdot 3.7^i\rceil$ points, which is smaller than $ 3.7^k$. The depth of the pair $\{x,y\}$ in $P'$ would be at least $\frac{\lceil4.7\cdot 3.7^i\rceil}{4.7}\ge 3.7^i$ and at most  $|P'|{-}\frac{|P'|}{4.7}{-}2\le\frac{3.7\lceil4.7\cdot 3.7^i\rceil}{4.7}{-}2<3.7^{i+1}$.

\vspace{10pt}\noindent{\bf Remark.}
Ramos and Via\~{n}a \cite{Ramos2009} conjectured that the correct bound in Lemma~\ref{Ramos-lemma} should be $\lfloor\frac{n}{4}\rfloor-1$. This bound, if true, would be the best achievable \cite{HaywardRappaportWenger:ExtremalResultsCirclesContainingPoints-1989}. Moreover, it would immediately improve the bound of Theorem~\ref{thm:disks-plane} to $3^k$.

\vspace{10pt}\noindent{\bf Remark.}
The proof of Theorem~\ref{thm:disks-plane} relies on the existence of a pair whose depth is large but also not too large (Corollary~\ref{Ramos-cor}). The latter part, however, is not given by Lemma~\ref{ball-depth-lemma}, and thus the proof of Theorem~\ref{point-ball-thr} relies only on the existence of a pair of large depth.

\section{Coloring Tuples of Grid Points}
In this section we prove Theorem~\ref{grid-rectangle-thr} and Theorem~\ref{grid-box-thr}. Our main tool is an application of the  Lov\'{a}sz local lemma \cite{Erdos1975}. Recall that $\Hyp$ is the family of all hypergraphs $H$ that is obtained by taking the vertex set  $P$ of a regular grid in $\Re^2$ setting $V(H){:=}P$ and $E(H){:=}\{r\cap P:r$ is an axis-parallel rectangle in $\Re^2\}$. We prove for any natural number $k\ge 2$ that $f_{\Hyp}(2,k)\leq \sqrt{ck\ln k}$, \rev{for some constant $c$}.

\begin{proof}[Proof of Theorem~\ref{grid-rectangle-thr}]
Set $m:=\sqrt{ck\ln k}$. Consider any hypergraph in $\Hyp$. It suffices to prove the statement for hyperedges (i.e.~rectangles) that contain at least $m$ and at most $2m$ points, because any rectangle with more than $2m$ points can be shrunk to have between $m$ and $2m$ points. Any polychromatic coloring with respect to shrunk rectangles is also polychromatic coloring with respect to original rectangles.

We prove the existence of a $k$-coloring of pairs of $P$ using the probabilistic method and the Lov\'{a}sz local lemma~\cite{Erdos1975}; the coloring itself can be computed by the constructive proof of the local lemma that is given by Moser and Tardos~\cite{Moser2010}. 
We color the pairs in $P$ by one of the colors from $\{0,\dots,k-1\}$ uniformly at random. 
Let $R$ be any rectangle containing at least $m$ and at most $2m$ points. The number of pairs in $R$ is at least ${m \choose 2}\ge{\frac{1}{4}ck\ln k}$. Consider any color $i\in\{0,\dots,k{-}1\}$. The probability that a pair $\{x,y\}$ is not colored $i$ is 
\[\PR(\{x,y\} \text{ is not colored } i)=1-\frac{1}{k}\le e^{\frac{-1}{k}}.\]
Since the pairs are colored independently, the probability that no pair in $R$ is colored $i$ is
\[\PR(\text{no color }i \text{ in } R)\le \left(e^{\frac{-1}{k}}\right)^{\frac{1}{4}ck\ln k}=\frac{1}{k^{c/4}}.\]
Summing over all the $k$ colors, the union bound implies that the probability of a ``bad'' event $E_R$ that the rectangle $R$ is not polychromatic (not having at least one of the $k$ colors) is
\[p=\PR(E_R)\le \frac{1}{k^{\frac{c}{4}{-}1}}.\]

Notice that for two rectangles $R$ and $R'$ that do not overlap on any  grid point the events $E_R$ and $E_{R'}$ are independent.

Fix a rectangle $R$. The length and width of $R$ is at most $2m$. Thus, $R$ lies in a square $S$ of side length $2m$. The rectangle $R$ can overlap only with rectangles that are at most $2m$ points away from $R$ in each direction. Thus, the overlapping rectangles must lie in a (bigger) square $S^+$ of side length $6m$. The number of points in $S^+$ is $(6m)^2$. Thus, the number of overlapping rectangles is at most $(6m)^4$ because each such rectangle can be defined by two points (say bottom left and top right points). Therefore, any bad event $E_R$ is independent of all but at most 
\[d=6^4m^4=6^4c^2k^2\ln^2 k\le 6^4c^2k^4\] other events.
We claim that  $4pd<1$. Indeed, notice that 
\[4pd \le 4\frac{1}{k^{\frac{c}{4}{-}1}}6^4c^2k^4=\frac{4\cdot 6^4c^2}{k^{\frac{c}{4}{-}5}}\le \frac{4\cdot 6^4c^2}{2^{\frac{c}{4}{-}5}},\]
where the last inequality is valid because $k\ge 2$. The last expression is a decreasing function \rev{over the interval $[126,+\infty)$, and its value is less than $1$ at $c=126$}. 
Since $4pd<1$, the local lemma implies that with a positive probability none of the bad events occur, which in turn implies the existence of a $k$-coloring of pairs in $P$ such that every rectangle containing at least $m$ points is polychromatic. This completes the proof.
\end{proof} 

\vspace{10pt}\noindent{\bf Remark.}
Due to our desire to have a short proof, we did not optimize the constant $c$ in Theorem~\ref{grid-rectangle-thr}. It can be improved in several ways, for instance, (i) one can shrink rectangles to have between $m$ and $m+\sqrt{m}$ points instead of $m$ and $2m$ points, (ii) one can consider the neighborhood of  $R$ itself instead of the neighborhood of the enclosing square $S^+$ which is larger, (iii) to count rectangles intersecting $R$ one can exclude those that contain less than $m$ points, and (iv) by using the stronger threshold of $epd<1$ for the local lemma~\cite{Shearer1985}.

For the proof of Theorem~\ref{grid-box-thr}, recall that the grid points are in dimension $d$, the hyperedges are defined by boxes (or balls), and we want to color $t$-tuples for $t\ge 1$.
\begin{proof}[Proof Sketch for Theorem~\ref{grid-box-thr}]We apply Lov\'{a}sz local lemma similar to what we did in the proof of Theorem~\ref{grid-rectangle-thr}. Due to the shrinkability of boxes (and balls), one can consider only those that contain $m=\Theta(\left(k\ln k)^{1/t}\right)$ points. We color the $t$-tuples of points by $\{0,\dots,k-1\}$ uniformly at random. Let $B$ be a box (or a ball) containing $m$ points. Then $B$ contains ${m\choose t}=ck\ln k$ tuples of size $t$, for some constant $c=c(d,t)$ {\color{mycolor}to be chosen}. Thus  
\[p=\PR(B \text{ is not polychromatic})\le \frac{1}{k^{c{-}1}}.\]
Observe that $B$ lies in a hypercube $S^+$ of side length $3m$ such that all boxes or balls that overlap $B$ lie in $S^+$. The number of points in $S^+$ is $(3m)^d$. Thus the number of overlapping boxes is at most $(3m)^{d^2}=O((k\ln k)^{d^2/t})$ as each box can be defined by $d$ points; for balls it is $(3m)^{d(d+1)}=O((k\ln k)^{d(d+1)/t})$ as each ball is defined by $d+1$ points. Then we can choose $c$ such that $4\cdot\frac{1}{k^{c{-}1}}\cdot(3m)^{d(d+1)}<1$ after which the claim follows by the local lemma.
\end{proof}

\section{Coloring Tuples in Hypergraphs of Bounded VC-Dimension}

In this section we first prove Theorem~\ref{lem:depth-boundedVC} --- the existence of a $(d+1)$-tuple of depth at least $\frac{n}{c}$, with  \milos{$c = 18(d+1)!$}, in any hypergraph $H$ of VC-dimension at most $d$ and with $n\ge 2d+2$ vertices, for every $d\ge 2$. This result \rev{is of independent interest}.
%, is analogous to Lemma~\ref{Ramos-lemma} and Lemma~\ref{ball-depth-lemma}. 
Our proof uses the probabilistic method, which is a common technique in combinatorics and discrete geometry \cite{Alon2008,clarkson1989applications}. We then use this result to prove our coloring result of Theorem~\ref{thm:poly-boundedVC}.

\begin{proof}[Proof of Theorem~\ref{lem:depth-boundedVC}]
Let $V$ be the vertex set of $H$ and $\E$ be the hyperedge set of $H$. Let $k$ be the maximum depth of a $(d+1)$-tuple in $V$ with respect to $H$. \milos{Our goal is to show that $k \ge n/c$.} 

If $H$ has a $(d+1)$-tuple that is not in any hyperedge of $\E$ then, by definition, $d_H=n-(d+1)\ge \frac{n}{2}$ and the lemma holds. Therefore, we assume that each $(d+1)$-tuple $T$ belongs to some hyperedge of $\E$. Let $e_T\in \E$ be a hyperedge that identifies the depth of $T$, i.e, $T\subseteq e_T$ and $d_H(T)=|e_T|-(d+1)$.

\rev{Observe that $k\neq 0$ because otherwise every $(d+1)$-tuple must be a hyperedge of $H$, in which case the VC-dimension is at least $(d+1)$, contradicting our assumption about the VC-dimension of $H$. For the rest of the proof we assume that $k\ge 1$.}

Let $X \subseteq V$ be a random sample of the vertices from $V$ where each vertex is chosen from $V$ independently \milos{with a fixed probability $p$, to be determined.}
%, satisfying $p>c/n$.} 
Let \milos{$\Y$} be the number of all hyperedges in $H[X]$ with cardinality equal to $d+1$.
Note that $|X|$ and $\Y$ are random variables.

Let $A$ be the event that $|X| \ge d+2$. Then  by the law of total expectation, conditioned over $A$, we have
$$\EXP( \Y ) = \EXP( \Y | A) \PR (A) + \EXP( \Y| \overline{A}) \PR (\overline{A}).$$
%We will show that $\EXP( \Y )\leq(np)^d$.
%In view of \eqref{exp1} it suffices to show that $\EXP( \Y|A )\le(np)^d$ and $\EXP( \Y|\overline{A} )\le(np)^d$.

If $|X| < d+2$ then there is at most 1 hyperedge of cardinality $d+1$ in $H[X]$, and thus $\EXP( \Y|\overline{A} )\PR (\overline{A})\le 1$. 
Hence, $\EXP( \Y ) \le \EXP( \Y | A) \PR (A) + 1$.

%Assume now that $|X| \ge d+2$.
    Note that the VC-dimension $d'$ of the projection $H[X]$ is at most $d$, and recall that $d\ge 2$. By the Sauer-Shelah-Perles lemma \cite{Buzaglo2013,Sauer1972,Shelah1972} we get $\Y \leq \sum_{i=0}^{d'} {\binom{|X|}{i}}$. Applying this inequality, we have
%\begin{equation}
%\notag
\begin{align*}
\milos{\EXP(\Y)} &\milos{\le \EXP\left(\sum_{i=0}^{d'} {\binom{|X|}{i}} \Big|A\right)\PR (A) +1}
    \\   &\milos{ \le\EXP\left(\sum_{i=0}^{d} {\binom{|X|}{i}} \Big| A\right)\PR (A)  +1}
%    \\ &\milos{ =\EXP\left(\binom{|X|}{0}+\binom{|X|}{1}+\binom{|X|}{2}+\dots+\binom{|X|}{d} \Big| A\right)\PR (A)  +1}
    \\ &\milos{ =\EXP\left(1+|X|+\frac{|X|(|X|{-}1)}{2!}+\dots+\frac{|X|(|X|{-}1)\cdots(|X|{-}d+1)}{d!} \Big| A\right)\PR (A)  +1}
    \\ &\milos{ =\EXP\biggl(|X|(|X|{-}1)\cdots(|X|{-}d+1) \biggl[\frac{1}{|X|(|X|{-}1)\cdots(|X|{-}d+1)}+\frac{1}{(|X|{-}1)\cdots(|X|{-}d+1)}}
    \\ & \hspace*{2cm}\milos{ +\frac{1}{2!(|X|{-}2)\cdots(|X|{-}d+1)}+\dots+\frac{1}{d!}\biggr] \Big| A\biggr)\PR (A)  +1}
    \\ &\milos{ \leq\EXP\biggl(|X|(|X|{-}1)\cdots(|X|{-}d+1) \biggl[\frac{1}{2d}+\frac{1}{2d}+\frac{1}{d}+\dots+\frac{1}{d}\biggr] \Big| A \biggr)\PR (A) +1} \\ &\hspace*{9cm}\milos{ \textrm{(since $d\geq 2$ and $|X|\geq d+2$)}}    \\   
    &\milos{=\EXP\Big(|X| (|X| {-}1)\cdots (|X| {-} d +1) \, |A\Big)\PR (A) +1}
    \\    & \milos{= \left(\sum_{i=d+2}^n i(i-1) \cdots (i-d+1) \frac{\binom{n}{i}p^i (1-p)^{n-i}}{\PR (A)}  \right) \PR (A) +1}\\ 
    & \milos{\le \sum_{i={d}}^n \frac{n(n-1)\cdots (n-i+1)}{(i-d)!}p^i (1-p)^{n-i} +1}
    \\ 
    & \milos{= n(n-1)\cdots (n-d+1) p^{d} \left(\sum_{i={d}}^n \frac{(n-d)\cdots (n-i+1)}{(i-d)!}p^{i-d} (1-p)^{(n-d)-(i-d)}\right)  +1}
    \\ 
    & \milos{\rev{= n(n-1)\cdots (n-d+1) p^{d} \left(\sum_{i={d}}^n {n-d\choose i-d}p^{i-d} (1-p)^{(n-d)-(i-d)}\right)  +1}}
    \\ 
    & \milos{\rev{= n(n-1)\cdots (n-d+1) p^{d} \left((p + (1-p) )^{n-d}\right)  +1}}\hspace*{1cm}\textrm{\rev{(by Newton's Binomial Theorem)}}
\\   & \milos{= n(n-1)\cdots (n-d+1) p^{d}+1}.
    \end{align*}
%\end{equation}

    For each $(d+1)$-tuple $T$ in $V$ the probability that $T$ is a hyperedge in $H[X]$ is at least the probability that $e_T \cap X=T$ which is exactly $p^{d+1}(1-p)^{d_H(T)}$.
    Each such probability is at least $p^{d+1}(1-p)^k$ because $d_H(T) \leq k$ for all $T$. By interpreting $\Y$ as the sum of binary indicator variables (one for each $(d+1)$-tuple) and then applying the linearity of expectation, we get %on the expectation $\EXP(|\Y|)$,
\begin{equation}
    p^{d+1}(1-p)^k {\binom{n}{d+1}} \leq \EXP(\Y). \label{e:prev}
\end{equation}
\milos{We can now set \rev{$p:=\frac{1}{2k}$}. Combining the inequality~(\ref{e:prev}) with the upper bound we obtained for $\EXP(\Y)$, we get}
    \[
    \milos{    (1/\rev{2k})^{d+1}(1-1/\rev{2k})^k {\binom{n}{d+1}} \leq n(n-1)\cdots (n-d+1) (1/\rev{2k})^{d}+1.}
    \]
 Knowing that $(1-1/\rev{2k})^k \geq \frac{1}{4}$ \rev{for any $k\ge 1$}, we have
    \[
    \milos{    \frac{1}{4} \frac{n-d}{(d+1)!} \leq \rev{2k}+\frac{\rev{(2k)}^{d+1}}{n(n-1)\cdots (n-d+1)},}
    \]
\milos{ and $d<n/2$ implies}
     \[
    \milos{    \frac{n}{8(d+1)!}  \leq \rev{2k}\left(1+\frac{\rev{(2k)}^d}{(n/2)^d}\right).}
    \]
\milos{If we assume that $k < n/c$, we get}
     \[
    \milos{    \frac{n}{8(d+1)!}  \leq \frac{\rev{2}n}{c}\left(1+\left(\frac{\rev{4}}{c}\right)^d\right),}
    \]
\milos{ which clearly does not hold for $c=\rev{18}(d+1)!$. Hence, we have $k \ge n/c$.}
\end{proof}

Now we use Theorem~\ref{lem:depth-boundedVC} to prove Theorem~\ref{thm:poly-boundedVC} which states that for every \rev{$d\ge 2$}, every shrinkable hypergraph $H$ with VC-dimension at most $d$ and with $n\ge 2d+2$ we have $f_H(d+1,k)\le c^{k-1}$, for any \rev{$k\ge 2$} and  \rev{$c = 18(d+1)!$}.

\begin{proof}[Proof of Theorem~\ref{thm:poly-boundedVC}]
Our proof uses an argument somewhat similar to that of Theorem~\ref{point-ball-thr}. We color each $(d+1)$-tuple of vertices of $H$ with one of the $k$ colors $\{0,1,\ldots,k-1\}$  as follows:
\[ \text{color of }T =
  \begin{cases}
        0   & \quad \text{if } d_H(T) = 0 \\
    i+1  & \quad \text{if } c^i\le d_H(T) < c^{i{+}1} \text{ for some }i \in \{0,1,\ldots,k-3\}\\
    k-1  & \quad \text{if } d_H(T) \ge c^{k-2}.
  \end{cases}
\]

We prove that any hyperedge $e$ of $H$ that contains at least $c^{k{-}1}$ vertices, has $(d+1)$-tuples of all the $k$ colors. Since $H$ is shrinkable there is a hyperedge $e_0$ such that $e_0\subset e$ and $|e_0|=d+1$. The hyperedge $e_0$ itself is a $(d+1)$-tuple of depth $0$, and hence is colored $0$. 

For each $i\in\{0,\ldots,k-2\}$ we show that there is a $(d+1)$-tuple with color $i{+}1$ in $e$. Let $e_i$ be a hyperedge  such that $e_i\subseteq e$ and $|e_i|=c^{i+1}$ ($e_i$ exists by shrinkability of $H$ and the fact that $c^{i+1}\le c^{k-1}\le |e|$). Let $V$ be the set of vertices in $e_i$. The VC-dimension of $H[V]$ is at most the VC-dimension of $H$, which is at most $d$. The number of vertices of $H[V]$ is $|V|=c^{i{+}1}\ge c\ge\ahmad{9(d+1)!}\ge 2d+2$. Hence by Theorem~\ref{lem:depth-boundedVC} the hypergraph $H[V]$ has a $(d+1)$-tuple $T$ of depth at least $\frac{|V|}{c}=c^i$. On the other hand, the depth of $T$ in $H[V]$ is at most $c^{i+1}-(d+1)$ because this is the most number of vertices (other than $T$) that an edge of $H[V]$ can have. Thus $c^i\le d_{H[V]}(T)< c^{i{+}1}$. We argue that these bounds also hold with respect to $H$, that is $c^i\le d_{H}(T)< c^{i{+}1}$. The lower bound is implied from the monotonicity of depth with respect to projections. The upper bound is verified by the edge $e_i$ itself, which is an edge of $H$, it contains $T$, and it has at most $c^{i{+}1}-(d+1)$ other vertices of $H$. Therefore, $T$ has color $i+1$.
\end{proof}

\subsection{Coloring Pairs of Points with respect to Pseudo-Disks} 

Both Theorem~\ref{thm:poly-boundedVC} and Theorem~\ref{lem:depth-boundedVC} are applicable to hypergraphs defined by points and pseudo-disks in the plane as they are shrinkable \cite{Pin14} and have VC dimension at most 4 \cite{Aronov2021}. In particular Theorem~\ref{thm:poly-boundedVC} implies that that $5$-tuples of points can be $k$-colored such that any pseudo-disk, with at least $c^{k-1}$ points, contains $5$-tuples of all the $k$ colors. In what follows we present a result that colors $2$-tuples (pairs of points).
   
    \begin{theorem}\label{thm:pseudo-disks}
        Let  $\D$ be a family of pseudo-disks and $P$ be a finite set of points in the plane. Then, the pairs of $P$ can be $k$-colored 
    so that any region in $\D$, containing at least $c^k$ points of $P$, also contains pairs from each of the $k$ colors, for some constant $c$.
    \end{theorem}
\begin{proof}[Proof Sketch] The proof follows similar ideas as in the proof of Theorem~\ref{point-ball-thr} and Theorem~\ref{thm:poly-boundedVC}. What we need is shrinkability of pseudo-disks, which is given in \cite{Pin14}, and the existence of pairs of linear depth, which we prove below in Theorem~\ref{thm:pseudo-disk-depth}.
\end{proof}

 An immediate application of this result is that for points in the plane and homothets of any convex body, we have $f(2,k) =O(c^k)$ for some constant $c$, which relates to the depth.

\begin{theorem}
    \label{thm:pseudo-disk-depth}
Let $\cal D$ be a family of pseudo-disks and let $P$ be a set of \rev{$n\ge 5$} points in the plane. Then, there \rev{exists} a pair $p,q$ of points in $P$ such that any region in $\cal D$ that contains $p$ and $q$ also contains at least $\frac{n}{c}$ other points of $P$ for some constant $c$; that is $d_P(p,q)\ge \frac{n}{c}$.
\end{theorem}
\begin{proof}
The proof is similar to that of Theorem~\ref{lem:depth-boundedVC}. For the sake of completeness we give the details.

If $P$ has a pair $(p,q)$ that is not contained in any region of $\cal D$, then, by definition, $d_P(p,q)=n-2$ and the lemma holds. Therefore, we assume that each pair $(p,q)$ of points in $P$ belongs to some region of $\cal D$. Let $k$ denote the maximum depth of a pair of points in $P$ with respect to $\cal D$. For every pair of points $p,q$ let  $d_{pq}\in \cal D$ be a ``witnessing" region  to the fact that the depth of $(p,q)$ is at most $k$, that is, $p\in d_{pq}$, $q\in d_{pq}$ and $|d_{pq} \cap P|-2 \leq k$. \rev{Observe that $k\neq 0$ because otherwise every pair of points should be separated by a closed Jordan region, and since we have at least five points the boundaries of a pair of regions must intersect more than two times, which contradicts $\cal{D}$ being a collection of pseudo-disks. For the rest of the proof we assume that $k\ge 1$.}

    Let $X \subseteq P$ be a random sample of the points of $P$ where each point is chosen from $P$ independently with a fixed probability $1/n\le p\le 1$ (to be determined later). Let $E_X$ be the set of all pairs $p,q \in X$ for which there exists a region $d \in \cal D$  containing only $p,q$ and no other points of $X$, i.e.,  $d \cap X = \{p,q\}$.
    Notice that $|X|$ and $|E_X|$ are random variables.
    It is implied from a result of \cite[Lemma 1]{Buzaglo2013} that the graph $(X,E_X)$ is planar and hence it has at most $3|X|$ edges. Thus
    $|E_X|\le  3 |X|$ and hence $\EXP(|E_X|)\le \EXP( 3|X|) = 3\EXP(|X|) =3n p$.
     
    For each pair $(p,q)$ in $P$ the probability that $(p,q)\in E_X$ is at least the probability that $d_{pq} \cap X=\{p,q\}$ which is exactly $p^{2}(1-p)^{|d_{pq}\cap P| -2}$.
    Each such probability is at least $p^{2}(1-p)^k$ because $|d_{pq}\cap P| -2 \leq k$. By interpreting $|E_X|$ as the sum of binary indicator variables (one for each pair of points in $P$) and then applying the linearity of expectation, we get %on the expectation $\EXP(|\E'|)$,
    \[
    p^{2}(1-p)^k {\binom{n}{2}} \leq \EXP(|E_X|).
    \]
    Combining this with the upper bound $3np$ and setting $p:=\frac{1}{\rev{2k}}$ we get
    \[
    {\binom{n}{2}} \leq \frac{\rev{6k}n}{ (1-\frac{1}{\rev{2k}})^k}.
  \]
    Since $(1-\frac{1}{\rev{2k}})^k \geq \frac{1}{4}$ \rev{for any $k\ge 1$}, we have
    \[
    {\binom{n}{2}} \leq \rev{24}kn.
    \]
    So $k \geq \frac{\binom{n}{2}}{\rev{24}n} \geq \frac{n}{c}$ for $c = \rev{96}$. This completes the proof. 
\end{proof}

\section{Relationships Between $f(1,k)$, $f(t,k)$, and $\epsilon$-$t$-Nets}
\label{relationship-section}
\rev{Recall from~\eqref{lowerbound} that $f(t,k)\ge \frac{1}{e}\cdot t k^{\frac{1}{t}}$.}
The following folklore result, though very simple, gives an upper bound on $f(t,k)$ in terms of $f(1,tk^{\frac{1}{t}})$. We have not seen a proof of this result elsewhere; we give a short proof.

\begin{lemma}
    \label{thm:poly-well-behaved}
    For any hypergraph $H$ and natural numbers $t$ and $k$ with $tk^{\frac{1}{t}}\ge 1$ it holds that $f_H(t,k)\le f_H(1,\lceil t k^{\frac{1}{t}}\rceil)$.
\end{lemma}
\begin{proof}
Set $x:=\lceil t k^{\frac{1}{t}}\rceil$ and color the vertices with $x$ colors such that any edge with $f_{H}(1,x)$ vertices is polychromatic. Let $C_1$ be the set of these $x$ colors. Now consider all combinations of $t$ different colors in $C_1$. We have ${x\choose t}$ many such combinations. Consider each combination as a new color class, and let $C_t$ be the set of the new colors (there is a one-to-one correspondence between the colors in $C_t$ and the combinations of $t$ colors in $C_1$). The standard lower bound for binomial coefficient yields \[{x \choose t}={\lceil t k^{\frac{1}{t}}\rceil  \choose t}\ge {t k^{\frac{1}{t}}  \choose t}\ge \left(\frac{tk^{\frac{1}{t}}}{t}\right)^t= k.\] Thus $C_t$ has at least $k$ colors. We use the colors in $C_t$ to color each  $t$-tuple $T$ of vertices of $H$. If the vertices of $T$ are colored by $t$ different colors in $C_1$, we color $T$ by the corresponding color in $C_t$. If the vertices of $T$ are colored by less than $t$ colors in $C_1$ (i.e., some vertices have the same color), we color $T$ by an arbitrary color in $C_t$. 

To prove that our coloring is polychromatic, consider any hyperedge $e$ with at least $f_H(1,x)$ vertices. Since $e$ is polychromatic with respect to the vertex coloring by $C_1$, it contains $x$ vertices of distinct colors. Of these vertices, each combination of size $t$ would be a $t$-tuple that is colored by a unique color in $C_t$. Based on this and the one-to-one correspondence introduced earlier, we conclude that all colors of $C_t$ appear in $t$-tuples of $e$, and hence $e$ is polychromatic with respect to $t$-tuple coloring by $C_t$.  
\end{proof}

%\todoin{Ahmad: Verify this. Which theorem of [1]?}

\begin{proof}[Proof of Theorem~\ref{cor:poly-well-behaved}]
A combination of \eqref{lowerbound} and Lemma~\ref{thm:poly-well-behaved} gives an asymptotically tight bound for $f_H(t,k)$ when $f_H(1,x)$ is linear in $x$. The statement of the theorem then follows by setting $x:=tk^{1/t}$.
\end{proof}

%\begin{corollary}    
%\label{cor:polyhalfplanes}    Let $\Hyp$ be the family of all hypergraphs defined by points and halfplanes. Then for any natural numbers $t$ and $k$ we have $f_{\Hyp}(t,k)= \Theta(k^{\frac{1}{t}})$.\end{corollary}

\subsection{Relationship to $\epsilon$-$t$-nets}
The existence of small-size $\epsilon$-nets was first shown by Haussler and Welzl~\cite{hw-ensrq-87} who proved that any finite hypergraph with VC-dimension $d$ has an $\epsilon$-net of size $O((d/\epsilon)\log(d/\epsilon))$, a bound that was later improved to $O((d/\epsilon)\log(1/\epsilon))$ in~\cite{KPW92}. The $\epsilon$-nets are studied extensively in computer science and have found applications in areas such as computational geometry, algorithms,
machine learning, and social choice theory; see, e.g.,~\cite{ABKKW06,AryaFM17,AFM18,BEHW89,Chan16,Chan18}.
Alon et al.~\cite{AJKSY22} showed the existence of small-size $\epsilon$-$t$-nets in various geometric hypergraphs and hypergraphs with bounded VC-dimension. 

Here, we prove Theorem~\ref{thm:tnet-decomposition} that gives a decomposition of $t$-tuples into $\epsilon$-$t$-nets. It shows the existence of $\Omega\left((\frac{\epsilon n}{t})^t\right)$ pairwise disjoint $\epsilon$-$t$-nets in a hypergraph $H$ with $n$ vertices for which $f_H(1,x)=O(x)$.

   % \begin{remark}        Note that the bound $\Omega(\epsilon^t n^t)$ in Theorem~\ref{thm:tnet-decomposition} is asymptotically optimal. Indeed, consider a hyperedge $S \in \E$ containing exactly $\epsilon n \geq f_{\Hyp}(k)=\Theta(k)$ vertices. Then $S$ must contain at least one $t$-tuple from each of the family of the $\epsilon$-$t$-nets in the decomposition. However, there are at most $\binom{\epsilon n}{t}= O(\epsilon^t n^t)$ of them.     \end{remark}

\begin{proof}[Proof of Theorem~\ref{thm:tnet-decomposition}]
From \rev{Theorem~\ref{cor:poly-well-behaved} 
%and linearity of $f_H(1,x)$
we have $f_{H}(t,k)= %f_H(1,tk^{\frac{1}{t}})=
\Theta(tk^{\frac{1}{t}})$}, for any $k$ where $tk^{\frac{1}{t}}\ge 1$. Let $k'$, with $tk'^{\frac{1}{t}}\ge 1$, be the largest integer such that $f_{H}(t,k') \leq  \epsilon n$ (notice that $f_{H}(t,k'{+}1) >  \epsilon n$). Then, every hyperedge of size at least $\epsilon n$ contains tuples of all color classes. This means that each of the $k'$ color classes is an $\epsilon$-$t$-net.  To estimate $k'$ we have $\epsilon n<f_{H}(t,k'{+}1)=O(t(k'{+}1)^{\frac{1}{t}})$. This would imply that $k' = \Omega\left((\frac{\epsilon n}{t})^t\right)$. When $t$ is a constant this gives $\Omega(\epsilon^tn^t)$.

To verify the tightness of this bound for constant $t$, consider a hyperedge $e$ with exactly $\epsilon n$ vertices. The number of distinct $t$-tuples in $e$ is ${{\epsilon n}\choose t}\le \left({\frac{e\cdot\epsilon n}{t}}\right)^t=O(\epsilon^tn^t)$. This means that we cannot decompose $t$-tuples into more than $O(\epsilon^tn^t)$ families of $\epsilon$-$t$-nets because otherwise $e$ would violate the $\epsilon$-$t$-net property of some color class. 
\end{proof}

\subsection{Cover-decomposability}
\label{cover-decomposable-section}
The case $k=2$ is usually referred to as bichromatic coloring. The notion of cover-decomposability is also derived from this case.
Recall that a hypergraph is cover-decomposable if $f(1,2)$ is bounded, and it is $t$-cover-decomposable if $f(t,2)$ is bounded.

Below we prove that $t+1\le f_H(t,2)\le\max\{f_H(1,2), t+1\}$ for any hypergraph $H$. This improves upon the previous upper bound of $\max\{f_H(1,2), 2t-1\}$, given in \cite[Proposition~4]{AckKesPalvo2021}. 

\begin{proof}[Proof of Theorem~\ref{thm:cover-decomposable-vs-tuples}]
The lower bound is obvious because in order for an edge to be polychromatic, it must contain at least two distinct $t$-tuples, and hence at least $t+1$ vertices ($t-1$ vertices could be shared between the two $t$-tuples).

Now we prove the upper bound. The upper bound holds if $f_H(1,2)$ is unbounded. Assume that $f_H(1,2)$ is bounded.  Let $V$ be the set of vertices of $H$, and let $\phi \colon V \rightarrow \{\text{red},\text{blue}\}$ be a two-coloring of the vertices in $V$ that achieves $f_H(1,2)$.
We use $\phi$ to obtain  a two-coloring $\xi$ of the $t$-tuples of $V$ such that $\xi\colon {\binom{V}{t}} \rightarrow \{\text{odd},\text{even}\}$, where $\binom{V}{t}$ denotes the set
of all $t$-element subsets of $V$. We color each $t$-tuple $T$ by the parity of the number of its blue vertices. If the number of blue vertices in $T$ is odd, then we color it odd, and if the number of blue vertices in $T$ is even, then we color it even. We show that this is a valid two-coloring of $t$-tuples, in the sense that any hyperedge $e\in H$ of cardinality at least $\max\{f_H(1,2), t+1\}$ contains two $t$-tuples such that one is colored odd and the other is colored even. 

Since $|e|\ge f_H(1,2)$, $e$ is polychromatic under $\phi$. Thus, there are two vertices in $e$, say $r$ and $b$, colored red and blue, respectively. Since $|e|\ge t+1$, $e{\setminus}\{r,b\}$ has at least $t-1$ vertices. Pick an arbitrary subset $t'$ of $e{\setminus}\{r,b\}$ of size exactly $t-1$. Then $t'\cup\{r\}$ and $t'\cup \{b\}$ are two $t$-tuples in $e$ and the cardinality of their blue vertices differs by $1$. Therefore, one of them is colored odd, and the other is colored even.
\end{proof}

\section{Conclusions}
We presented several new results on polychromatic coloring of tuples in hypergraphs. In particular we gave an exponential upper bound (in terms of the number of colors) for polychromatic coloring of \rev{$t$-tuples of points with respect to balls in dimension $d$  for $t=\lfloor\frac{d+3}{2}\rfloor$ (Theorem~\ref{point-ball-thr}). Our proof benefits from the existence of deep $t$-tuples. Similar bounds can be obtained for larger values of~$t$ because deep $t$-tuples still exist. However, deep $t$-tuples may not exist for any smaller value of~$t$~\cite{BaranySchmerlSidneyUrrutia:PointsBallsEuclideanSpace-1989}, and as such it would be interesting to obtain bounds for smaller values of $t$}. For coloring pairs of points in the plane with respect to disks we presented a better bound (Theorem~\ref{thm:disks-plane})---such a bound does not exist for coloring each point individually. \ahmad{The best lower bound that we are aware of is just the trivial linear bound.} One natural question is whether a polynomial upper bound is achievable for coloring pairs.

Another natural question is to close the gap between $\Omega(\sqrt{k})$ and $O(\sqrt{k\ln k})$ for polychromatic coloring of pairs of grid points with respect to axis-aligned rectangles. \rev{
A pair of grid points in the plane can be mapped to a point of the 4-dimensional grid. This means that a vertex-coloring for the 4-dimensional grid would imply a pair-coloring for the grid in the plane. It would be interesting to verify whether this connection could give a better bound. See~\cite{Bukh2021} and references therein for results about coloring points with respect to cuboids in higher dimensions.
}

It would also be interesting to explore similar bounds for other point sets and other objects, for example, one may relate the bounds to the fatness of objects.

\section{Acknowledgement}
This work was initiated at the 20th Gremo's Workshop on Open Problems, held in Wergenstein, Switzerland, in June 2023. It was then continued at the 11th Annual Workshop on Geometry and Graphs, held at the Bellairs Research Institute in Holetown, Barbados, in March 2024. We thank the organizers and participants of both workshops.
%We also thank the reviewers of SoCG 2025 who verified our proofs and provided valuable feedback.

\rev{We are very grateful to the anonymous reviewers who meticulously
verified our proofs, and provided valuable feedback that improved the accuracy and readability of our results.}

\bibliographystyle{plainurl}
\bibliography{references.bib}

\end{document}